\documentclass[11pt]{amsart}

\usepackage{amsfonts,amssymb,stmaryrd,amscd,amsmath,latexsym,amsbsy}
\usepackage{graphicx}
\usepackage{amssymb}
\usepackage{amsfonts}
\usepackage{latexsym}

\newtheorem{theorem}{Theorem}[section]
\newtheorem{lemma}[theorem]{Lemma}

\newtheorem{corollary}[theorem]{Corollary}
\theoremstyle{definition}

\newtheorem{example}[theorem]{Example}

\newtheorem{remark}[theorem]{Remark}

\newcommand{\ben}{\begin{enumerate}}
\newcommand{\een}{\end{enumerate}}

\theoremstyle{plain}

\newtheorem*{sol}{Solution}

\theoremstyle{definition}

\theoremstyle{remark}

\newcommand{\solu}[1]{\begin{sol}{\bf (\ref{#1})}}

\begin{document}

\title{On the breakup of air bubbles in Hele-Shaw cell}

\author{Vladimir Entov}
\address{V.E.: Deceased}

\author{Pavel Etingof}
\address{P.E.: Department of Mathematics, Massachusetts Institute of
Technology, Cambridge, MA 02139, U.S.A.}
\email{etingof@math.mit.edu}

\maketitle

\section{Introduction}

The problem of contraction of an air bubble in a Hele-Shaw cell filled with a 
Newtonian fluid under the influence of suction of air from the
bubble has been intensively studied by many physicists and
mathematicians for more than 25 years, see \cite{GV} and
references therein. In particular, in \cite{EE1}, the authors 
suggested an analytic theory which allows one 
to give a complete description of the asymptotics of the 
bubble shape as its area goes to zero, and, in
particular, to find the point of its contraction. In a number of
applications of the contraction problem, for instance, in the
theory of gas recovery, the following question, only briefly 
discussed in \cite{EE1}, is important: will the air bubble
fall apart during contraction, or will it remain connected 
until all the air has been extracted? In this paper, we study
this question in detail. In particular, we propose some
sufficient conditions of breakup of the bubble, and ways 
to find the contraction points of its parts. In the theory of gas
recovery, these points are interpreted as the optimal positions of
the gas-producing wells. 

We note that all the results of this paper
(with the exception of explicit solutions) 
extend to flows in a curved Hele-Shaw cell, 
along the lines of \cite{EE2,EV}. 

The structure of the paper is as follows. In Section 2, we
describe the mathematical model of the contraction problem, and 
recall some of the results from \cite{EE1} on the connection
between the dynamics of the bubble and its gravity potential.
We also study the asymptotic shapes of bubbles contracting 
to a degenerate critical point of the potential. 
In Section 3, we define the points of partial contraction, which
are contraction points of the bubbles which 
appear as a result of breaking of the initial bubble, 
and contract before the full contraction occurs,
and extend to them the results about points of complete
contraction from \cite{EE1}. In Section 4, we give 
a sufficient condition of breakup of a symmetric bubble;
this is an extension of a result from \cite{EE1}. In Section 5, 
we consider the process of regulated contraction, 
which is simultaneous contraction of two or more bubbles 
under prescribed rates of extraction from each bubble. 
In Section 6, we consider the case of a bubble which breaks up
into two bubbles under contraction, and discuss the question
whether one can use regulated contraction to make these two
bubbles contract simultaneously (development of singlularities in
solutions may be a problem). A strategy of extraction 
which allows one to do so is called a synchronizing strategy, 
and we study such strategies in some detail. In particular, in
Section 7 we study the asymptotics of contraction under a
synchronizing strategy, and show that, like in the case of a
single bubble, the two bubbles contract at critical points of 
the potential (generically, nondegenerate local minima),
and the bubble shapes generically tend to ellipses, 
whose axes are determined by the eigenvalues of the Hessian 
of the potential at the minima. In Section 8, we characterize
domains that are on the boundary between those that admit a
synchronizing strategy and those that don't: we show that the
potential of such a domain should have a degenerate critical
point of the potential;  we also study the asymptotics 
of contraction to such a point. In Section 9 we characterize domains that are on
the boundary between those that break up and those that don't;
generically they develop an instantaneous $5/2$-cusp in the process of
contraction. In Section 10 we correct some computational errors
in the previous publications \cite{EE1,EV}.

{\bf Acknowledgments.} 
The work of the second author was  partially supported by the NSF grant
DMS-0504847. He thanks B. Gustafsson and D. Jerison 
for useful discussions. 

{\bf Note.} Sadly, the first author of this paper, Vladimir
M. Entov, passed away on April 10, 2008. 

\section{The mathematical model and its main properties}

\subsection{The model}
Let us recall the formulation of the problem and the main results
from the paper \cite{EE1}. Consider contraction of an air bubble
in an unbounded Hele-Shaw cell, which is filled with a Newtonian
fluid, under suction of air from the bubble. Let $B(t)$ be the
air domain at a time $t$. We assume that it is connected and
bounded, with (say) a smooth boundary. 
Its law of evolution in time is as
follows. In the fluid domain $B(t)^c$ (the complement of $B(t)$), 
which we assume to be connected, 
there is a potential vector field of fluid velocities, 
$v(x,y,t)=\nabla \Phi(x,y,t)$. The potential $\Phi$ is determined
at any time $t$ as a solution of the boundary value problem
\begin{equation}\label{e1}
\frac{\partial^2\Phi}{\partial
x^2}+\frac{\partial^2\Phi}{\partial y^2}=0,\ (x,y)\in B(t)^c;
\Phi|_{\partial B(t)}=0;\ \Phi(x,y)=-\frac{q}{2\pi}\log(r)+O(1),
r\to \infty,
\end{equation} 
where $r=(x^2+y^2)^{1/2}$, and $q>0$ is the rate of
suction. The velocity of the boundary is then equal 
to the velocity of the fluid particles on the boundary: 
\begin{equation}\label{e2}
v_b=\frac{\partial \Phi}{\partial n}.
\end{equation}

The contraction problem is to find the family of domains $B(t)$ 
for a given initial shape of the bubble $B(0)$ and given
rate $q$ of suction. It can be shown that for every 
initial domain with a smooth boundary, 
the solution of this problem exists and is unique 
on some interval of time $[0,\tau)$, $\tau>0$. 
It has an obvious monotonicity property $B(t_1)\subset B(t_2)$
for $t_1>t_2$. 

In a similar way one can define contraction of several
bubbles. In this case, the domain $B(t)$ is a union
of finitely many disjoint simply connected domains (bubbles). 

\subsection{Weak solutions}
Equations (\ref{e1},\ref{e2}) determine contraction 
as long as it reduces to continuous deformation of the boundaries
of the bubbles. However, in the process of contraction, the
boundary $B(t)$ may undergo topological transformations. 
For instance, parts of the boundary can collide (Fig.1). 

In this case, the above definition of the contraction does not apply
as it is, and needs clarification. 

Namely, let $S$ be the area of $B(0)$, and $t^*=S/q$ be 
the time of complete extraction of the air. Let us call a
family of bounded domains $B(t)$, $t\in [0,t^*)$, a {\it
weak solution} of the contraction problem if for every
$t\in [0,t^*)$, $B(t)$ is a disjoint union of finitely many
simply connected domains, so that 

(i) $B(t_1)\subset B(t_2)$ for $t_1>t_2$;

(ii) the area of $B(t)$ is $S-qt$; and

(iii) there exists a closed set of times $T\subset [0,t^*]$
(of topological transformations) containing $0$, $t^*$ 
such that for any $t\in T$, a small enough interval
$(t,t+\varepsilon)$ does not intersect $T$, and on the intervals of time
$(\tau_1,\tau_2)$ not intersecting with $T$ the domain $B(t)$
is a ``classical'' solution of the contraction problem, i.e. is
defined by equations (\ref{e1},\ref{e2}). 

In other words, a weak solution is ``glued'' from
usual (classical) solutions at points $\tau\in T$ where 
the domain $B(t)$ undergoes
topological transformations. It is natural to assume that 
it describes the actual process of contraction in the case when 
the breakup of the bubble does occur. 

It is known (see \cite{GV}) that the contraction problem has a unique weak
solution. An example of a weak solution is given on Fig.1; from this
example one can see two types of topological transformations that
can occur during contraction: 

1) breakup of a bubble into two pieces, and 

2) disappearance of a bubble.

As a result of these transformations, the number of bubbles
changes in the process of contraction. 

\begin{remark} It can be shown that the set $T$ is finite.
However, a detailed proof of this would be long and we do not
give it here. 
\end{remark}

\subsection{The gravity potential}
Let us define the gravity potential of a bounded domain $B$ to be
the function 
\begin{equation}\label{e3}
\Pi_B(\xi,\eta)=\frac{1}{2\pi}\int_B \log |z-\zeta|dxdy,
\end{equation}
where $\zeta=\xi+i\eta, z=x+iy$. 
This function satisfies the Poisson differential equation in
$\Bbb R^2$ with logarithmic asymptotics at infinity:
\begin{equation}\label{e4}
\Delta \Pi_B=\chi_B,\
\Pi_B(\xi,\eta)=\frac{S}{2\pi}\log|\zeta|+O(1),\
\zeta\to\infty,
\end{equation}
where $\chi_B(\xi,\eta)$ is the characteristic function of $B$,
and $S$ is the area of $B$. 
It is shown in \cite{EE1} that this function is closely related 
to the contraction problem. Namely, we have the following
theorem. 

\begin{theorem}\label{t1} Let $B(t)$ be a weak solution to the
contraction problem. Then 

(i) The gravity potential inside $B(t)$ 
changes by a constant under contraction:
\begin{equation}\label{e5}
\Pi_{B(0)}(\xi,\eta)-\Pi_{B(t)}(\xi,\eta)=C(t),\ (\xi,\eta)\in
B(t). 
\end{equation}

(ii) Let us set $\Phi(x,y,t)=0$ if $(x,y)\in B(t)$. 
Then we have 
\begin{equation}\label{e6}
\Pi_{B(0)}(\xi,\eta)=K-\int_0^{t^*}\Phi(\xi,\eta,t)dt.
\end{equation}
\end{theorem}

\begin{example}
Let us say that a simply connected bounded domain $B$ is
algebraic of degree $d$ if its Cauchy transform 
$h_B=-\bar z+4\partial_z \Pi_B$
(which is analytic in $B$) is actually a rational function of
degree $d-1$ (see \cite{Gu,EV}). In this case the same is true for $B(t)$ for $t>0$,
and thus it may be shown, similarly to \cite{Gu}, that the
boundary of $B(t)$ is defined by the algebraic equation $Q(z,\bar
z)=0$, where $Q$ is a polynomial of degree $2d$.
The genus of this algebraic curve is thus $\le (2d-1)(d-1)$, 
and thus the number of components of $B(t)$ is at most 
$(2d-1)(d-1)$.  
\end{example}

\subsection{Points of complete contraction} 
Let us say that a point of $\Bbb R^2$ is a point of complete
contraction if it belongs to $B(t)$ for all $t\in [0,t^*)$. 
Thus the points of complete contraction are the points of
disappearance of the bubbles which ``survive'' until the time of
complete contraction $t^*$. The set of all points of complete 
contraction is the intersection $\cap_{t<t^*}B(t)$. 
In \cite{EE1} we described the structure of this set. 

\begin{theorem} (i) The set of points of complete contraction is
finite. 

(ii) The points of complete contraction are the global minimum
points of the gravity potential. 
\end{theorem}

\begin{remark}
Note that if the boundary of a domain $B$ is smooth, then 
by the Hopf's strong maximum principle, 
a global minimum of $\Pi_B$ cannot be attained on the boundary
of $B$.
\end{remark} 

\subsection{Asymptotics of contraction} 

It is shown in \cite{EV} that when a bubble contracts completely to a point,
and the Hessian of the potential at that point is nondegenerate, 
then the boundary of the bubble has the asymptotic shape of an
ellipse, whose half-axes are directed along 
the eigenvectors of the Hessian, and 
their lengths are inverse proportional to its eigenvalues. Here
we would like to extend this result to the case when the Hessian
may be degenerate. 

Namely, assume that the bubble $B$ contracts at a point $0$. 
In this case, $\Pi_B(0)$ is an isolated global minimim point
of $\Pi_B$, \cite{EE1,EV}. Let us assume that $0$ is a degenerate critical
point, and the kernel of the Hessian of $\Pi_B$ at zero is the
x-axis, i.e. $\Pi_B(x,y)=\frac{1}{2}y^2+O(|z|^3)$ near $0$. 
For simplicity let us first assume that the bubble is
symmetric with respect to the x-axis. 
Then 
$$
\Pi_B=\frac{1}{2}y^2+\frac{\beta}{2n} x^{2n}+...,
$$
where $\beta$ is some positive number,
and ... are monomials strictly inside the Newton polygon.  
Let us call $n$ the degree of the critical point $0$. 

Let $B_*(t)$ be image of the bubble at the time
$t$ under the renormalization 
$x\to cx$, $y\to c^{2n-1}y$, where 
$c=c(t)$ is chosen in such a way that 
the diameter of $B_*(t)$ is $2$ (so $c(t)$ behaves like 
$(t^*-t)^{-1/2n}$ as $t\to t^*$).     

Define the polynomials 
$$
Q_n(u)=\sum_{k=0}^n
\frac{(2k)!}{4^kk!^2}u^{n-k}.
$$
 
\begin{theorem}\label{symbub}
The boundary of the domain $B_*(t)$ tends to the curve 
$$
y^2=\beta^2(1-x^2)Q_{n-1}^2(x^2).
$$  
\end{theorem}

\begin{proof} 
We write 
$\Pi_B$ in the the form 
$$
\Pi_B=-\frac{1}{8}(z-\bar z)^2+\beta{\rm Re}(z^{2n}/2n)+...,
$$ 
where $...$ are the terms strictly inside the Newton polygon. 

The conformal map of the unit 
disk into the outside of $B(t)$ which maps $0$ to $\infty$ has the
form
$$
f(\zeta)=A(\zeta+\zeta^{-1})+\phi(\zeta), 
$$
where $\phi$ is an odd holomorphic function in the unit disk, 
and one may assume that $A>0$. 
Here $\phi=\phi_t$, $A=A(t)$, and $A\to 0$ as $t\to t^*$. 

Let $f^*$ denote the function obtained from $f$ by conjugating 
the coefficients of the Taylor series, 
and let $h_B(z)$ be the Cauchy transform 
of $B$: 
$$
h_B=\bar z-4\partial_z\Pi_B=z-2\beta z^{2n-1}+O(z^{2n}).
$$  
Then by Richardson's theorem (see \cite{Ri,EV}), the function 
$$
f^*(\zeta^{-1})-h_B(f(\zeta))
$$
extends holomorphically from the boundary of the unit disk to its
interior, and vanishes at $0$. Therefore, we
obtain, for $|\zeta|=1$:
$$
{\rm Im}\phi(\zeta)=-2\beta
A^{2n-1}{\rm Im}(\zeta+\zeta^{-1})^{2n-1}_++O(A^{2n}),
$$
where the subscript $+$ stands for the holomorphic part
(i.e. the nonnegative degree terms of the Taylor series). 

Now we claim that 
$$
{\rm Im}(\zeta+\zeta^{-1})^{2n-1}_+=4^{n-1}Q_{n-1}(\cos^2\theta)\sin\theta, 
$$
where $\zeta=e^{i\theta}$. This is proved easily by induction in $n$. 
Also, if $|\zeta|=1$ then $\zeta+\zeta^{-1}$ is real. 
Thus, we get for $|\zeta|=1$:
$$
y={\rm Im}f(\zeta)=-\beta
(2A)^{2n-1}Q_{n-1}(\cos^2\theta)\sin\theta+O(A^{2n}).
$$
On the other hand, we have 
$$
x={\rm Re}f(\zeta)=2A\cos\theta+O(A^2).
$$
Therefore, upon rescaling $x\to x/2A,y\to y/(2A)^{2n-1}$ 
we will obtain the following equations: 
$$
y=-\beta Q_{n-1}(\cos^2\theta)\sin\theta+O(A),
$$
$$
x=\cos\theta+O(A),
$$
Thus when $A$ goes to zero (i.e. for $t\to t^*$), we get 
the limiting shape 
$$
y^2=\beta^2 (1-x^2)Q_{n-1}^2(x^2),
$$
as desired. 
\end{proof}

Now let us consider the general case, i.e. 
a bubble which is not necessarily symmetric
with respect to the x-axis. In this case, 
$$
\Pi_B=\frac{1}{2}y^2+\alpha yx^n+\frac{\beta}{2n}x^{2n}+...=
$$
$$
-\frac{1}{8}(z-\bar z)^2+{\rm Re}(i\alpha
\frac{z^{n+1}}{n+1}+\beta\frac{z^{2n}}{2n})+...,
$$ 
where $\alpha,\beta$ are real and $\beta-n\alpha^2>0$
(to insure that $0$ is an isolated minimum of $\Pi_B$), and 
$$
h_B=z+2i\alpha z^n-2\beta z^{2n-1}+O(z^{2n}).
$$  
Applying a similar method to the one used in the 
symmetric case, we obtain the
following result. 

Let $B_*(t)$ be image of  
the bubble at the time
$t$ under the map $z\to z+i\alpha z^n$, 
followed by the renormalization 
$x\to cx$, $y\to c^{2n-1}y$, where 
$c=c(t)$ is chosen in such a way that 
the diameter of $B_*(t)$ is $2$.     

\begin{theorem}\label{nonsymbub}
The boundary of the domain $B_*(t)$ tends to the curve 
$$
y^2=(\beta-n\alpha^2)^2(1-x^2)Q_{n-1}^2(x^2).
$$  
\end{theorem}

The most interesting case for applications is $n=2$. 
This case corresponds to contraction of 
(symmetric) domains that are on the boundary between 
those that break up and those that don't. 

\begin{example}
Assume that $\Pi_B=\frac{y^2}{2}+\beta {\rm Re} z^4/4$.
This potential corresponds to the contracting bubble 
whose conformal map from the unit disk to the outside region 
has the form 
$$
f_t(\zeta)=A\zeta^{-1}+\frac{A}{1+6\beta A^2}\zeta-2\beta
A^3\zeta^3, 
$$
where $A=A(t)>0$ is some function.
(This map is found from the singularity correspondence, \cite{EV}). For
small enough $A$, this function is univalent and 
defines a bubble. Contraction of the bubble corresponds to
decreasing $A$ to $0$, leaving $\beta$ fixed. 
Then the bubble contracts to the origin with the asymptotic shape
given by the above theorem for $n=2$. This domain is on the
boundary between rupturing and non-rupturing domains. 
\end{example}

\begin{remark}
This analysis of asymptotic shapes is similar to the analysis 
of the shapes of the necks of bubbles during break-off 
which is done in \cite{LBW}.  
\end{remark}

\section{Points of partial contraction}

\subsection{Definition and properties of points of partial contraction}
Let $B_0(t)\subset B(t)$ be a connected component of the air
domain (i.e. a single bubble), which exists on the interval of
time $(\tau_f,\tau_c)$; namely, we assume that $\tau_f$ is the 
time of formation of the bubble $B_0$, and $\tau_c$ is the 
time of its disappearance (contraction). A point contained in
$B_0(t)$ for all $t\in (\tau_f,\tau_c)$ will be called a {\it
point of partial contraction.} In Fig.1, $P$ is a point of
complete contraction, and $Q$ is a point of partial contraction. 
A point of either complete or partial contraction will be called
a contraction point. To every contraction point there corresponds
a time of contraction $\tau_c$. 

\begin{theorem}\label{t4} (i) Every component $B_0$ of the air domain which
contracts without breakup contains a unique contraction point. 

(ii) A contraction point at a time $t$ is an (isolated) global
minimum point of the potential $\Pi_{B(0)\setminus B(t)}$, and vice
versa. In particular, the number of contraction points is finite. 
\end{theorem}

\begin{proof} By formula (\ref{e6}), 
$$
\Pi_{B(0)\setminus
B(t)}(\xi,\eta)=K_t-\int_0^t\Phi(\xi,\eta,\tau)d\tau,
$$
where $K_t$ is a constant. Since $\Phi(\xi,\eta,\tau)\le 0$, 
and $\Phi(\xi,\eta,\tau)=0$ if and only if $(\xi,\eta)$ is
contained in the closure of $B(t)$, we have 
$\Pi_{B(0)\setminus B(t)}\ge K_t$, and 
$\Pi_{B(0)\setminus B(t)}=K_t$ if and only if 
$(\xi,\eta)\in \cap_{\tau<t}B(\tau)$. Thus, if 
$(\xi,\eta)$ is a contraction point at a time $t$,
then $\Pi_{B(0)\setminus B(t)}(\xi,\eta)=K_t$, i.e. the potential
achieves its minimal value at $(\xi,\eta)$, and vice versa. 

On the other hand, the contraction points of the component $B_0$
are contained in the domain $B(0)\setminus B(t)$, where the
potential $\Pi_{B(0)\setminus B(t)}$ satisfies the Poisson
equation $\Delta\Pi=1$, i.e. is a real analytic function.
Thus the set of contraction points of $B_0$ is analytic
(as it is a connected component of the set of solutions of the equation
$\Pi_{B(0)\setminus B(t)}=K_t$). Also, it is compact and simply
connected. This implies that it consists of one point (see
\cite{EE1}). This point is thus an isolated point of global
minimum of the potential. On the other hand, it is clear that any
global minimum point of $\Pi_{B(0)\setminus B(t)}$ is a
contraction point. The theorem is proved. 
\end{proof}

\subsection{Finding points of partial contraction} 
If the bubble breaks up into two parts, which contract 
without further breakup, then Theorem \ref{t4} allows 
us to find the partial contraction point explicitly, 
provided that the derivative of the conformal map 
$f_0(\zeta)$ from the unit disk to the complement
of the initial domain $B(0)$ is a rational function. Indeed, let $\tau$ be the
moment of contraction of the bubble that contracts sooner, and $z_0$ be its
contraction point. The domain $B(\tau)$ is connected, so there exists
a conformal map $f_\tau(\zeta)$ of the unit disk into the complement
of $B(\tau)$, which also has rational derivative. It can be found
as described in \cite{EV}. Assume that the map $f_\tau$ is
known. Then by a direct computation one finds the potential
$\Pi_{B(\tau)}$ on the whole plane. Next, the point $z_0$ 
is found from the condition 
$$
\nabla(\Pi_{B(0)}-\Pi_{B(\tau)})(z_0)=0.
$$ 
Moreover, the value of the potential $\Pi_{B(0)}-\Pi_{B(\tau)}$
at $z_0$ must coincide with the value of this potential inside
the domain $B(\tau)$ (this value is constant inside $B(\tau)$ 
by Theorem \ref{t1}). This is a condition on the unknown time
$\tau$ of contraction.

A similar method allows one to find the contraction point
of the smaller bubble in the problem of simultaneous contraction 
of two circular bubbles, $K_1=\lbrace{z; |z|<R\rbrace}$ and 
$K_2=\lbrace{z; |z-a|<r\rbrace}$, $a,R,r>0$, 
$R>r$, $a>R+r$ (Fig.2). 

In this case, as was shown by P. P. Kufarev
\cite{Ku}, the conformal map of the unit disk to the complement
of the domain $B(\tau)$ at the time $\tau$ of partial contraction
has the form 
\begin{equation}\label{e9}
f_\tau(\zeta)=\frac{\beta\alpha^{-1}}{1-\alpha\zeta}+\frac{\gamma}{\zeta},
\ \gamma=\frac{1}{2}\left(a\alpha+\frac{2r^2+
R^2-\frac{qt}{\pi}}{a\alpha}\right),
\end{equation}
$$
\beta=\frac{1-\alpha^2}{2}\left(a\alpha-
\frac{2r^2+R^2-\frac{qt}{\pi}}{a\alpha}\right)
$$
where $\alpha>0$ and $\alpha^2$ is the middle root of
the cubic equation 
\begin{equation}
2a^4x^3-(2a^2(R^2-qt/\pi)+a^4)x^2+(2r^2+R^2-qt/\pi)^2=0,
\end{equation}
and from the equations 
\begin{equation}
\nabla \Pi_{B(0)\setminus B(\tau)}(z_0)=0,\ \Pi_{B(0)\setminus
B(\tau)}(z_0)=\Pi_{B(0)\setminus B(\tau)}|_{B(\tau)}
\end{equation} 
one finds the contraction point $z_0$ of the smaller bubble, 
and the time $\tau$ of partial contraction. 

\subsection{Asymptotics of partial contraction} 

As in the case of complete contraction, if the contraction point
is a nondegenerate minimum of the potential, then at times close 
to the time $\tau$ of partial contraction, the boundary of the
bubble has an approximate shape of an ellipse, whose half-axes
are directed along the eigenvectors of the Hessian of
$\Pi_{B(0)\setminus B(\tau)}$ at the contraction point, and 
their lengths are inverse proportional to its eigenvalues. This fact is
proved analogously to the case of complete contraction
(\cite{EV}).

It is interesting to study the rate of partial contraction, 
in the case when the bubbles contract at different times. 
Namely, assume we have two bubbles $B_1$ and $B_2$, and $B_2$
contracts to a point $P_2$ at a time $t'$, at which $B_1$ assumes
the shape of a domain $E$. We assume that $P_2$ is a
nondegenerate minimum of the potential $\Pi_{B(0)\setminus
B(t')}$. Let us conformally map the outside 
of $E$ onto the unit disk, so that $\infty$ maps to $0$, and the 
point $P_2$ to a point $b\in (0,1)$. Obviously, such $b$ and the
map are unique. Let $\zeta$ be the complex coordinate in the
disk; then the potential $\Phi$ at a time $t<t'$ close to $t'$
has the form 
$$
\Phi=-\frac{q}{2\pi}\log|\zeta|+\frac{Q}{2\pi}\log\vert
\frac{\zeta-b}{1-\zeta b}\vert+O(t'-t), 
$$
where $Q=Q(t)$ is the rate of contraction of the bubble $B_2$;
it is clear that $Q(t)\to 0$ as $t\to t'$. 

Let $\Gamma(t)$ be the image of the boundary of $B_2(t)$ in the
disk, and let $A(t)$ be the area of $B_2(t)$. 
Since the boundary of $B_2(t)$ for $t$ close to $t'$ is 
almost elliptic, the distance from the points of 
$\Gamma$ to the point $b$ is sandwiched between $c_1A(t)^{1/2}$
and $c_2A(t)^{1/2}$, where $c_1,c_2$ are some constants. 
On the other hand, $\Phi$ must vanish on $\Gamma(t)$. 
This yields
$$
\frac{Q}{2}\log A\to q\log b, t\to t'.
$$
Thus, $Q(t)$ is equivalent to $2q\frac{\log b}{\log A(t)}$ as
$t\to t'$. 

This allows us to determine the asymptotic behavior 
of $A(t)$ as $t\to t'$. To do so, let us introduce the variable
$\tau=t'-t$. Then asymptotially $A$ behaves as 
the solution of the the differential equation
$$
\frac{dA}{d\tau}=2q\frac{\log b}{\log A},
$$
with the initial condition $A=0$ as $\tau=0$. 
Solving this equation, we get the solution which is implicitly
defined by the equation 
$$
A\log A-A=2q\tau \log b. 
$$
So we obtain the following result. 

\begin{theorem}\label{logslow} 
We have 
$$
A\sim \frac{2q\tau\log b}{\log\tau},\ Q\sim \frac{2q\log
b}{\log\tau},\ \tau\to 0. 
$$
\end{theorem}

This result can be generalized to the case when the bubble $B_2$
contracts to a degenerate minimum. Namely, assume that 
the contraction point is a critical point of degree $n$. 
Then, conducting a similar asymptotic analysis, using the results
of subsection 2.5, we obtain the following theorem. 

\begin{theorem}\label{logslow1} 
We have 
$$
A\sim \frac{2nq\tau\log b}{\log\tau},\ Q\sim \frac{2nq\log
b}{\log\tau},\ \tau\to 0. 
$$
\end{theorem}

Thus we see that when $t$ is close to $t'$ then 
the contraction of the bubble $B_2$ is logarithmically slow, 
and almost all air is extracted from $B_1$. 

\section{Sufficient conditions for breakup of symmetric bubbles}

\begin{theorem}\label{t5} 
Assume that the gravity potential of a simply connected domain
$B(0)$ symmetric with respect to a point $P$ (respectively, a
line $\ell$) achieves a global minimum at a point $Q\ne P$ 
(respectively, $Q\notin \ell$). Then $B(0)$ breaks up in the
process of contraction. 
\end{theorem}

\begin{proof} The symmetric point $Q'\ne Q$ is also the global minimum
point for the gravity potential. By Theorem \ref{e3}, the points
$Q$ and $Q'$ are points of complete contraction of the domain
$B(0)$. Therefore, $B(0)$ must break up. 
\end{proof} 

\begin{corollary} (see \cite{EE1}) Let $B(0)=\lbrace{(x,y);
|x|<b, y^2<f(x)\rbrace}$, where $f(x)$ is an even piecewise
smooth function, nonnegative on $[-b,b]$, such that $f(b)=0$ (Fig.3). 

Then $B(0)$ breaks up in the process of contraction if 
\begin{equation}\label{e12}
\int_0^b\frac{d(xf(x)^{1/2})}{x^2+f(x)}>\pi/2. 
\end{equation}
\end{corollary}

\begin{proof}
If (\ref{e12}) holds, then 
$$
\frac{\partial^2 \Pi_{B(0)}}{\partial x^2}(0)=
\frac{1}{2}-\frac{1}{\pi}\int_0^b\frac{d(xf(x)^{1/2})}{x^2+f(x)}<0,
$$
so the origin is not a global minimum point of the potential. 
Hence the global minimum is attained at another point. Because of
the central symmetry, the domain $B(0)$ must break up. 
\end{proof}

\begin{theorem}\label{t6} (see \cite{EV}, problem 2 on p.32)
Let $B(0)$ be a simply connected domain, symmetric with respect
to the horizontal axis, and the function 
$\pi(x):=\Pi_{B(0)}(x,0)$ has more than one local extremum 
on $(-\infty,\infty)$. Then the domain $B(0)$ breaks up 
in the process of contraction.  
\end{theorem}

\begin{proof} Assume the contrary, i.e. that the domain does not
break up. Let $x_1(t)<x_2(t)$ be the intersection points of the
boundary $\partial B(t)$ with the horizontal axis (Fig. 4) 

By the monotonicity property of contraction, the function
$x_1(t)$ is increasing, and the function $x_2(t)$ is decreasing
on $[0,t^*)$, and $\lim_{t\to t^*}x_1(t)=\lim_{t\to
t^*}x_2(t)=x_0$. The rays $[x_2(t),+\infty)$ and
$(-\infty,x_1(t)]$ are flowlines of the flow, and the flow is
directed from infinity, so the potential $\Phi(x,0,t)$ increases
from $-\infty$ to $0$ on the interval $(-\infty,x_1(t)]$, equals 
zero on $(x_1(t),x_2(t))$, and decreases from $0$ to $-\infty$ on
$[x_2(t),+\infty)$. Thus, if $\xi_1\ge \xi_2\ge x_0$ or 
$\xi_1\le \xi_2\le x_0$, then for any $t\in (0,t^*)$ one has
$\Phi(\xi_1,0,t)\le \Phi(\xi_2,0,t)$. Moreover, if this inequality
turns into an equality for all $t$, then $\xi_1=\xi_2$. Since 
$\Pi_{B(0}(\xi,\eta)=K-\int_0^{t^*}\Phi(\xi,\eta,t)dt$, 
these arguments imply that the function $\Pi_{B(0)}(x,0)$ is
strictly increasing on $(x_0,+\infty)$, and strictly decreasing
on $(-\infty,x_0)$, i.e. its unique local extremum is a minimum at
the point $x_0$. A contradiction. 
\end{proof} 

\section{Regulated contraction}

\subsection{Definition of regulated contraction}
Consider contraction of a domain which breaks up into two parts
at a time $\tau\in [0,t^*)$. After the time $\tau$, the process
of contraction may be controlled, creating different pressures 
inside the two components of the air domain by regulating the
amount of air which is pumped out of each component. This is a
generalization of the problem from Section 2. In particular, 
it is interesting whether one can regulate contraction in such 
a way that both bubbles contract at the same time; in this case,
as we will see, the contraction points have the convenient
property that they are critical points of the potential 
of the initial bubble, and thus can be easily found.  
Let us consider this generalized problem in more detail. 

Consider the process of contraction of a system of two bubbles;
we don't assume that they were obtained as a result of breakup of
a single bubble. Assume that at a time $t\in [0,t^*)$, the air domain $B(t)$
consists of the components $B_1(t)$ and $B_2(t)$, and the air is
pumped from $B_1(t)$ at the rate $q_1(t)$, and from $B_2(t)$ at 
a rate $q_2(t)$. This means that the velocity potential
$\Phi(x,y,t)$ is a solution of the boundary value problem
$$
\frac{\partial^2\Phi}{\partial
x^2}+\frac{\partial^2\Phi}{\partial y^2}=0,\ (x,y)\in B(t)^c;\
\Phi|_{\partial B_i(t)}=\Phi_i(t), i=1,2;\ 
$$
$$
\Phi(x,y)=-\frac{q_1(t)+q_2(t)}{2\pi}\log(r)+O(1),
r\to \infty,
$$
and the constants $\Phi_1(t)$ and $\Phi_2(t)$ are chosen in such
a way that 
\begin{equation}\label{e14}
\int_{\partial B_i}\frac{\partial \Phi}{\partial
n}d\ell=q_i(t),i=1,2,
\end{equation}
It is useful to extend $\Phi$ to the interior of $B(t)$: 
$\Phi=\Phi_i(t)$ in $B_i(t)$ for $i=1,2$. Then $\Phi(x,y,t)$ is 
an everywhere continuous function. The velocity of motion of the
boundaries $\partial B_1$ and $\partial B_2$ is $\frac{\partial
\Phi}{\partial n}$. The motion can be considered up to the time
of disappearance of one of the bubbles (we assume that 
there is no topological transformations of the first kind, i.e.,
formations of new bubbles). The contraction process in this
situation will be called {\it regulated contraction}
(as opposed to free contraction, i.e. with equal pressures in the bubbles). 
The vector-function $(q_1(t),q_2(t))$ will be called {\it the strategy
of air extraction}. 

\begin{theorem} \label{t7} The gravity potential inside
every component of the contracting domain changes by a constant 
in the process of regulated contraction (the constants may be
different for different components). 
\end{theorem}

This theorem is proved analogously to part 1 of Theorem \ref{t1}.

\begin{theorem} \label{t8}
Let $B_1^s,B_2^s$ be two smooth families of simply 
connected domains, $s=(s_1,s_2)$, $B_1^s\cap B_2^s=\emptyset$. 
If the potential of the system $\lbrace{B_1^s,B_2^s\rbrace}$
changes by a constant inside $B_1^s$ and inside $B_2^s$ under 
the change of $s$, and the areas of $B_1^s$ and $B_2^s$ are
independent of $s$, then $B_1^s$ and $B_2^s$ don't change under
the change of $s$.  
\end{theorem}

Thus, a pair of simply connected domains is locally uniquely
determined by the potential inside them, given up to additive
constants, and by the areas of the domains.  

\begin{proof}
Let $v_s: \partial B_1^s\cup \partial B_2^s\to \Bbb R$ is the
velocity of the boundary under the change of $s$. For any
continuous function $u(x,y)$, 
$$
\frac{d}{ds}\int_{B_1^s\cup B_2^s}udxdy=\int_{\partial B_1^s\cup
\partial B_2^s}uv_sd\ell.
$$
Thus, the conditions of the theorem imply that the function $v_s$
is orthogonal on $\partial B_1^s\cup \partial B_2^s$ to the
following functions: 

(i) the characteristic functions $\xi_j(z)$ of the boundary
comnponents $\partial B_j^s$ (because the areas of $B_j^s$ are
constant in $s$); and 

(ii) $\eta_1^w(z)={\rm Re}(w-z)^{-1}$, $\eta_2^w(z)={\rm
Im}(w-z)^{-1}$, where $w\in B_1^s\cup B_2^s$ (because of the
fact that the gradient of the gravity potential
is constant in $s$). These functions are dense in $L^2(\partial
B_1^s\cup \partial B_2^s)$. Hence, $v_s$ is identically zero. 
The theorem is proved. 
\end{proof}

\begin{theorem} \label{t9}
Let $\bar q_s(t)$, $t\in (\tau,\theta)$, $\bar q=(q_1,q_2)$, 
$s\in [0,1]$, be a smooth family of strategies of extraction
such that 
$\int_\tau^\theta \bar q_s(t)dt=\bar Q$ is independent of $s$. 
Let $B^s(\theta)=B^s_1(\theta)\cup B^s_2(\theta)$ be the result of
extraction according to the strategy $\bar q_s$ from the same 
initial domain $B(\tau)=B_1(\tau)\cup B_2(\tau)$. 
Then $B^s(\theta)$ is independent of $s$. 
\end{theorem}

\begin{proof}
By Theorem \ref{t7}, the gravity potential $B^s(\theta)$ is independent
of $s$ (and equals the potential of $B(\tau)$) up to an additive
constant (in each connected component). The areas of the domains
$B_1^s(\theta)$ and $B^s_2(\theta)$ are also independent of $s$
and equal $S_1-Q_1$, $S_2-Q_2$, respectively, where $S_j$ are the
areas of the components $B_j(\tau)$, and $Q_j$ are the
coordinates of the vector $\bar Q$, i.e., the volumes of the air
extracted from the first and the second bubble. By Theorem
\ref{t8}, $B^s_j(\theta)$ do not change under the change of $s$,
as desired.  
\end{proof}

Thus, the result of contraction depends only on the total
quantities of air extracted from the bubbles for a given 
period of time, and does not depend on other parameters of the
strategy. In other words, the transformations in the space of
domains defined by extraction of air from the first and the second
bubble, respectively, commute with each other. This is an analog 
of Richardson's result \cite{Ri} on the commutativity of injection
operations at different points (see also \cite{EV}). 

\subsection{The phase rectangle and the accessibility region}
Theorem \ref{t9} shows that the domains which can be obtained
from $B(0)$ under regulated contaction can be visualized by 
points of the ``phase'' rectangle
$0\le X\le S_1$, $0\le Y\le S_2$, where $S_1$ and $S_2$ are areas
of $B_1(\tau)$ and $B_2(\tau)$, respectively: to the point
$(X,Y)$ corresponds the domain which is obtained by extraction of
the volumes $S_1-X,S_2-Y$ from the first and second bubble,
respectively. The strategy of extraction is depicted by a path
inside the rectangle, which emanates from the corner $(S_1,S_2)$
(Fig. 5). 

It is important to note that in the process of regulated
contraction, pieces of the boundary may move 
towards the fluid region. If this happens for a particular
strategy, the problem of regulated contraction
is ill-posed in the vicinity of this strategy. 
Thus in general one has neither the monotonicity
property nor the existence of a weak solution up to the time of
complete contraction. In other words, the solution of the problem 
of regulated contraction may not exist for some strategies of
extraction. More specifically, in a generic situation solutions develop 
singularities in the following way: at some point of the boundary
one of the two bubbles develops a semicubic cusp directed
towards the fluid region (Fig.6), and after this time the solution
cannot be continued. 

The development of such a cusp was first
discovered by P. Ya. Polubarinova-Kochina \cite{Ko} for the problem
of contraction of the boundary of the oil region. The fact that
the singularity generically takes the shape of a semicubic cusp
is related to the fact that while the solution exists, the
boundary of the bubble is an analytic curve,
and a semicubic cusp is the simplest
singularity of a nonselfintersecting analytic curve. 

It follows from the above that not all paths in the phase
rectangle correspond to actual solutions, but only those that
lie in some region $\Omega$, which is the set of all points of
the rectangle that can be accessed by a strategy of extraction
in which both bubbles exist all the way up to the last moment. 
We will call $\Omega$ {\it the
accessibility region}. The boundary of the accessibility region 
consists (in the generic situation) of 
parts of the boundary of the phase rectangle, and 
curves, whose points correspond to pairs of bubbles, one of which
has a cusp (Fig.7) 

The path $\gamma$ corresponding to free
contraction divides the accessibility region into two parts. Namely, if the
pressure in the first bubble is kept higher than in the second bubble,
then the corresponding path lies below $\gamma$, and if it is kept
lower then the path lies above $\gamma$. Similarly, one can
construct the trajectory $\gamma_P$ of free contraction, starting from 
a domain corresponding to any point $P\in \Omega$. 
The accessibility region $\Omega$ is foliated by such
trajectories (Fig.8), which implies that $\Omega$ is
contractible (i.e., simply connected).  

If the initial domain consists of two symmetric bubbles, 
then the accessibility region is symmetric with respect to the
diagonal $Y=X$ of the phase square, 
and the path $\gamma$ of the free contraction is 
this diagonal. 

\begin{remark}
Theorem \ref{logslow} implies that the trajectories $\gamma_P$ of free
contraction which do not end in the origin are tangent to the
boundary of the phase rectangle at the endpoint, and the tangency is of the type 
$Y-a=cX\log(1/X)$ if the endpoint is $(0,a)$, 
and $X-a=cY\log(1/Y)$ if the endpoint is $(a,0)$ (for some
$c>0$). 
\end{remark} 

\section{Synchronizing strategies of extraction}
 
Let us say that a strategy $\bar q(t)$ is synchronizing for the
system of bubbles $B_1,B_2$, if the extraction according to this strategy
leads to simultaneous contraction of the bubbles $B_1$ and $B_2$ 
to a point. The path in the accessibility region which 
corresponds to a synchronizing strategy ends in the origin. 
Obviously, a synchronizing strategy exists iff the accessibility
region contains the origin. A domain $B(0)$ 
which breaks up under contraction 
into two bubbles which have this
property, will be called {\it synchronizable}. 

Under extraction of air according to synchronizing strategy $\bar q(t)$, the
bubbles $B_1$ and $B_2$ simultaneously contract to points $P_1$ and $P_2$. 
This means that the points $P_1$ and $P_2$ are limiting positions
of the boundaries $\partial B_1(t)$ and $\partial B_2(t)$ when
the time $t$ tends to the time $t^*$ of contraction. These points
are easily found from the shape of the initial domain.

\begin{theorem}\label{t10}
The contraction points $P_1$ and $P_2$ are critical points of the
potential $\Pi_{B(0)}$ (if they both belong to $B(0)$). 
\end{theorem}

\begin{proof}
\begin{lemma}\label{le}
Let $G$ be a bounded domain of area $S$. Then 
\begin{equation}
|\nabla \Pi_G(x,y)|\le \sqrt{S/\pi}, x,y\in \Bbb R.
\end{equation} 
\end{lemma}

\begin{proof}
Let $z=x+iy,w=u+iv$. Let $K$ be the disk of area $S$ centered at
$(x,y)$. Then we have 
$$
|\nabla
\Pi_G(x,y)|=\frac{1}{2\pi}\vert\int_G\frac{(z-w)dudv}{|z-w|^2}\vert\le 
\frac{1}{2\pi}\int_G\frac{dudv}{|z-w|}\le 
$$
$$
\frac{1}{2\pi}\int_K\frac{dudv}{|z-w|}=\frac{1}{2\pi}\int_0^{2\pi}d\rho
\int_0^{\sqrt{S/\pi}}d\lambda=\sqrt{S/\pi}, 
$$
as required. 
\end{proof}

Now we prove the theorem. Let $t_n$ be a sequence of times which
tends from below to $t^*$. Let $P_1^{(n)},P_2^{(n)}$ be sequences
of points in $B_1(t_n)$ and $B_2(t_n)$, which converge to $P_1$
and $P_2$ as $n\to \infty$. By Theorem \ref{t7}, we have 
$$
\nabla \Pi_{B(0)}(P_j^{(n)})=\nabla \Pi_{B(t_n)}(P_j^{(n)}).
$$
By Lemma \ref{le}, $|\nabla \Pi_{B(t_n)}(P_j^{(n)})|\to 0$ as
$n\to\infty$, as the area of $B(t_n)$ tends to $0$ for
$n\to\infty$. Hence, $\nabla \Pi_{B(0)}(P_j^{(n)})\to 0$,
$n\to\infty$, $j=1,2$. Since the gravity potential is a
$C^1$-function, this implies that 
$\nabla \Pi_{B(0)}(P_j)=0$, as desired. 
\end{proof}

Clearly, one of the contraction points $P_1,P_2$ coincides 
with the point $P$ of complete contraction in the sense of 
Section 1 (namely, the contraction point for the bubble that
contracts later under free contraction). 
This point, as we mentioned, is the global minimum
point of the gravity potential. Regarding the second point, it is
shown below that it is either a local minimum point 
or a degenerate critical point, and the degenerate critical point
arises for initial domains which lie on the boundary between
synchronizable and nonsynchronizable domains in the space of
domains. 

Let us call an initial domain $B(0)$ strictly synchronizable if the
corresponding accessibility region contains 
a sector 
$$
\lbrace{(X,Y)| X\ge 0,Y\ge 0,
X^2+Y^2<\varepsilon\rbrace}
$$ 
for sufficiently small
$\varepsilon$. Clearly, strict synchronizability is an open
condition, i.e. this is a property which is stable under small
deformations. Non-strictly synchronizable domains form the
boundary between synchronizable and nonsynchronizable domains. 
To illustrate this, we show in Fig.9,a,b,c what the accessibility region
looks like for a strictly synchronizable, non-strictly
synchronizable, and nonsynchronizable domain. 

\begin{theorem}\label{t11} If a domain $B(0)$ is strictly
synchronizable, then the contraction points $P_1$, $P_2$ 
for a synchronizing strategy are local minima of the potential
$B(0)$ (if they lie in $B(0)$). 
\end{theorem}

\begin{proof}
If the domain $B(0)$ is strictly synchronizable, then one of the
trajectories of free contraction ends in the origin (Fig.9).
 
Therefore, there exists a synchronizing strategy which
corresponds to free contraction on the interval
$(t^*-\varepsilon,t^*)$ for some $\varepsilon>0$. Then by Theorem
\ref{t4}, the points $P_1,P_2$ are points of global minimum 
of the potential $\Pi_{B(t^*-\varepsilon)}$. Since inside
$B_i(t^*-\varepsilon)$, 
$i=1,2$, the potential $\Pi_{t^*-\varepsilon}$ coincides with
$\Pi_{B(0)}$ up to constants, we see that $P_1$ and $P_2$ are
points of local minimum of $\Pi_{B(0)}$. 
\end{proof}

\begin{remark}
If the point $P_1$ or $P_2$ is outside the domain $B(0)$ 
(a priori, such a situation cannot be ruled out because of the failure of the
monotonicity property), then it is a local minimum point
of the analytic continuation of the potential of of the domain $B(0)$ from its
inside to its outside along the track of the corresponding bubble.   
\end{remark}

\section{Asymptotics of regulated contraction under a
synchronizing strategy}

\begin{theorem}\label{t12}
Let $B(t)=(B_1(t),B_2(t))$ be the evolution of the air domain
after its breakup, under extraction of air according to a
synchronizing strategy $\bar q(t)$, $P_1,P_2\in B(0)$ be the
contraction points, and $A_1,A_2$ the Hessian matrices of the
potential $\Pi_{B(0)}$ at these points. Let $\widehat{B_j(t)}$ be
the domains obtained from $B_j(t)$ by dilation $\lambda_j(t)$
times, where $\lambda_j(t)$ is chosen in such a way that 
$\widehat{B_j(t)}$ has a fixed diameter $d$. Then if the matrices
$A_1,A_2$ are nondegenerate, then the boundaries $\partial
\widehat{B_i(t)}$ tend to ellipses, whose axes are directed 
along eigenvectors of $A_1,A_2$, and the lengths of the half-axes
are inverse proportional to the eigenvalues of these matrices. 
\end{theorem}

\begin{proof} We prove the theorem for $B_1$; the proof for $B_2$
is the same. Assume that the point $P_1$ is the origin. 
Obviously, for $t$ close to $t^*$ the potential $\Pi_{\widehat{B_1(t)}}$
has a nondegenerate local minimum at some point $a(t)$ inside
$\widehat{B_1(t)}$. 
Let $E(t)=B_1(t)-a(t)$ be the domain obtained by translating 
$\widehat{B_1(t)}$ by the vector $-a(t)$. 
It can be seen, along the lines of \cite{EV}, 
that the boundary $\partial E(t)$ 
converges to some curve $\Gamma$.

Let $A_D(z)$ be the matrix 
of second derivatives of the gravity 
potential of a domain $D$.
It is easy to see that 
$A_{E(t)}(z)\to A_{B(0)}(0)=A_1$, $t\to t^*$,
if $z\in E(t)$ for $t$ close to $t^*$.
This implies that the potential of the domain 
$E$ bounded by the curve $\Gamma$ is a
quadratic function, whose Hessian matrix is $A_1$, and which has
a minimum at zero, i.e. it is $(A_1z,z)+C$. 
By Sakai's theorem, this implies (see e.g. \cite{EV}) that $E$ is an ellipse,
whose axes are directed along the eigenvectors of $A_1$, and 
the lengths of half-axes are inverse proportional to its eigenvalues.
The center of the ellipse is situated at the origin. 
The theorem is proved. 
\end{proof} 

\section{Properties of the potential of a non-strictly
synchronizable domain}

\subsection{Contraction at degenerate critical points} 
\begin{theorem}\label{t13}
Let $B(0)$ be a synchronizable domain,
$P_1$ and $P_2$ be points of its contraction
under a synchronizing strategy, and $A_1,A_2$ be matrices of second
derivatives of the potential at the points $P_1,P_2$. 
In this case, if the matrices $A_1,A_2$ are nondegenerate, then the domain
$B(0)$ is strictly synchronizable. 
\end{theorem}

\begin{proof}
Assume the contrary, i.e. that the domain is not strictly
synchronizable. Consider the synchronizing strategy, which
corresponds to a path $\gamma$ in the accessibility region
$\Omega$ (Fig.10). 

We may assume that the path $\gamma$ goes along 
the boundary of $\Omega$. In this case, the boundary of one of
the bubbles has a persistent singularity in the process 
of contraction. But by Theorem
\ref{t12}, the limiting shape of this bubble is an
ellipse, which does not have a singularity. 
This is a contradiction.  
\end{proof}

\begin{corollary}\label{co}
If a domain is synchronizable, but not strictly synchronizable, 
then the matrix $A_2$ of second derivatives of the potential at
the point $P_2$ is degenerate (here $P_1$ is the point of
complete contraction, and $P_2$ is the second contraction
point). 
\end{corollary}

In the generic situation this means that the potential 
has at $P_2$ a critical point of type saddle-node, i.e. 
in some Cartesian coordinates 
$$
\Pi_{B(0)}(x,y)=\frac{y^2}{2}+\beta \frac{x^3}{3}+..., 
$$
where $...$ stand for the terms inside the Newton polygon. 

\subsection{Asymptotics of contraction}

It is easy to show that in the situation Corollary \ref{co}, 
the limiting curve $\Gamma$ is a line interval
(slit) of length $d$ (traveled forward and backward), 
in the direction of the kernel of the matrix $A_2$.
However, it is interesting to study a finer asymptotics of 
contraction in this situation. Consider the generic case when 
the contraction point of the singular bubble $B_2(t)$ is $0$, 
and the potential at this point is as above: 
$$
\Pi_{B(0)}(z)=-\frac{(z-\bar z)^2}{8}+\frac{\beta{\rm Re}(z^3)}{3}+...
$$

Let $B_*(t)$ be image of the bubble $B_2(t)$ at the time
$t$ under the renormalization 
$x\to cx$, $y\to c^2y$, where 
$c=c(t)$ is chosen in such a way that 
the diameter of $B_*(t)$ is $2$.     

\begin{theorem}
(i) The boundary of $B_*(t)$ tends to the curve
$$
y^2=\beta^2\left(x+\frac{1}{2}\right)^3\left(x-\frac{3}{2}\right).
$$

(ii) The bubble $B_1(t)$ contracts at some point $a<0$ of the
real axis. Thus the cusp of $B_2$ is always directed precisely towards $B_1$
at the time of contraction. 
\end{theorem}

\begin{proof}
The method of proof is the same as the one we used for the asymptotic
analysis of a bubble contracting to a degenerate minimum. Namely,
we have 
$$
h_{B_2(t)}(z)=K(t)+z-2\beta z^2+...,
$$
where $...$ are negligible terms as $t\to t^*$. 
(Note that $h_{B_2(t)}(z)$ depends on $t$ because of the presence
of the second bubble $B_1(t)$). 
Thus, the conformal map of the unit disk into 
the outside of $B_1(t)$ which maps $0$ to $\infty$ looks like 
$$
f(\zeta)=A(\zeta+\zeta^{-1})+B\zeta^2+C\zeta+D+...
$$
Since $f^*(\zeta^{-1})-h_D(f(\zeta))$ is holomorphic and vanishes
at infinity, we get 
$$
B=-2\beta A^2, C=-4\beta AD,
$$
modulo negligible terms. Also, we have a cusp in the ``saddle'' direction of the
singularity, i.e. in our case in the negative direction. 
Thus, $f_t'(-1)=0$, which gives 
$2B-C=0$ modulo negligible terms. This means that modulo negligible terms
we have 
$$
D=A,\ C=-4\beta A^2.
$$
This implies that 
$$
x={\rm Re}f=2A(cos\theta+\frac{1}{2})+O(A^2), 
$$
and 
$$
y=-2\beta A^2(\sin 2\theta+2\sin\theta)+O(A^3). 
$$
Thus after the rescaling $x\to x/2A$, $y=y/4A^2$ and sending 
$A$ to zero we get the limiting curve 
$$
y^2=\beta^2\left(x+\frac{1}{2}\right)^3\left(x-\frac{3}{2}\right),
$$
and part (i) of the theorem follows. 
We also get $K=6\beta A^2+O(A^3)>0$ for small $A$, 
which implies that the second bubble disappears at some point
$a<0$, hence (ii).
\end{proof}

We see from the proof of this theorem that 
when $A$ goes to zero, the area of $B_1$ goes down as $c_1A^2$ and
the area of $B_2$ as $c_2A^3$. This shows that the trajectory
corresponding to our strategy (i.e. the upper boundary of the
accessibility region) behaves near the origin as a semicubic 
parabola $Y=c X^{3/2}$ (Fig.9(b)). 

\subsection{A sufficient condition of breakup} 
Corollary \ref{co} gives us the following sufficient condition of
breakup. 

\begin{corollary}\label{co1}
Suppose $B^s$, $s\in [0,1]$, is a smooth family of simply connected domains, 
such that $B^0$ is symmetric (centrally or axially), 
and the potential of $B^0$ admits a global minimum point in $B^0$
which is not fixed by the symmetry. Assume that the analytic
continuation of the potential of $B^s$ for any $s$ 
does not have degenerate critical points. Assume also that it is
known that $B^s$ undergo at most one topological transformation
under contraction for all $s$. Then the domains $B^s$ 
break up under contraction for all $s$.
\end{corollary}

\begin{proof}
By Theorem \ref{t5}, $B^0$ is synchronizable (and 
the synchronizing strategy is free extraction). 
If some $B^s$ does not break up, then it is not synchronizable, 
so for some $\sigma<s$
$B^\sigma$ is synchronizable, but not strictly 
synchronizable. By Corollary \ref{co}, the analytic continuation
of $\Pi_{B^\sigma}$ has a degenerate critical point. 
Contradiction. 
\end{proof}

\section{The boundary between rupturing and non-rupturing
domains} 

Having multiple local minima of the gravity potential is not a
necessary condition for the breakup of a bubble. Indeed, it is
obvious that a domain in Fig.11 breaks up in the process of
contraction, although its potential has a unique critical point
(so this domain is not synchronizable). 

In this connection, it is
interesting to study domains which lie on the boundary between
rupturing and non-rupturing domains. 

\begin{theorem}\label{t15}
Suppose that $B$ is a generic domain, which is on the boundary
between rupturing and non-rupturing domains. Then $B$ does not
break up under contraction, but at some time $t$ the domain $B(t)$
has a cusp on the boundary, locally equivalent to $y=x^{5/2}$. 
This singularity disappears immediately under further
contraction. 
\end{theorem}

In other words, on the boundary between the sets of 
rupturing and non-rupturing domains in the space of
($C^k$-smooth) domains, a dense open set is formed by 
domains that have the property stated in the theorem. 

\begin{proof} (sketch) Consider a smooth family of simply
connected bounded domain of generic position, $B^s$, $s\in
[0,1]$. Assume that for $s\le \sigma$ the domain $B^s$ does not
break up, while for $s>\sigma$ it does. 
Let $\tau(s)$ be the time of breakup of the domain $B^s$ for
$s>\sigma$. Let $\tau=\lim_{s\to \sigma}\tau(s)$.
One can show that in the situation of
general position this limit exists and is not equal to
zero. Consider the family of curves
$\Gamma(s)$, $s\in (\sigma,1]$, which are obtained from $B^s$
by contraction during the time $\tau(s)$ (where
$\tau(\sigma):=\tau)$). 

The curves $\Gamma(s)$ for $s>\sigma$ have a simple
self-tangency at some point. It is easy to see that typical
degenerations of such curves into simple closed curves 
have the structure described in the theorem: for $s=\sigma$ at
the point of disappearance of the loop there forms a cusp of
degree $5/2$. Namely, a typical such family is 
$$
y^2=x^4(x+\varepsilon), \varepsilon=\varepsilon(s),
$$ 
where $\varepsilon(\sigma)=0$, and $\varepsilon(s)>0$ for $s>\sigma$. 
\end{proof}

\begin{remark} 
The appearance of the singularity in the process of contraction 
is at first sight a strange phenomenon, as the contraction
problem has good properties of correctness and stability.
Nevertheless, solutions with such type of instantaneous 
singularities do exist. They were first discovered by Howison in
\cite{Ho}. More precisely, he showed that in the process of
contraction there can appear instantaneous cusps of degrees
$(4n+1)/2$, $n\ge 1$, while cusps of degree $(4n-1)/2$ cannot
appear. For $n=1$, these are the instantaneous cusps of degree
$5/2$ that we have just considered. Solutions with these
properties form a subset of codimension 1 in the space of all
solutions. 
\end{remark}

If we restrict ourselves to polynomial domains of degree $\le n$,
then the boundary between rupturing and non-rupturing domains, as
follows from Theorem \ref{t15}, is a piece of 
an algebraic surface in the space of coefficients. For small
degrees $n$ this equation is not difficult to write down. If two
domains can be connected by a curve not intersecting this
surface, then either they both break up under contraction, or
they both don't. 

\section{Corrections to \cite{EE1,EV}}

We use the opportunity to correct some errors in \cite{EE1,EV}. 

1. The formula for the gravity potential of the ellipse 
$\frac{x^2}{a^2}+\frac{y^2}{b^2}=1$
given in \cite{EE1,EV} (\cite{EE1},p.517, and \cite{EV}, (4.14))
is incorrect. The correct formula is 
$$
\Pi=\frac{1}{2}(\frac{b}{a+b}x^2+\frac{a}{a+b}y^2)+C(a,b).
$$
(i.e., $a$ and $b$ need to be switched). The same correction
needs to be made in formulas (24),(25) of \cite{EE1}, and the 
two sentences after formula (25). Hence the correct condition 
of division of the bubble for $a=2b$ is $c<\sqrt{3}b$. The same 
corrections should be made in \cite{EV} in the example on p.52,
and in the answer to problem 4 on p.68.  

2. As a result of 1, Theorem 4.10 in \cite{EV} should say that 
the lengths of the half-axes of the limiting ellipse are {\it
inverse} proportional to the eigenvalues of the Hessian matrix. 
The same correction is to be made on p.535 of \cite{EE1}
(the power $-2$ of the Hessian matrix should be replaced with
$+2$).

3. The left hand side of formula (19) should read
$2\frac{\partial^2\Pi_B}{\partial x^2}(0,0)$ 
(the factor of 2 is missing). 
The same correction should be made in the 
first formula on p.52 of \cite{EV}. 
The right hand side of the formula in Theorem 6.4 in \cite{EV} 
should be $\pi$, not $\pi/2$. In the first formula on 
p.52 of \cite{EV}, the factor $2/\pi$ should be $1/\pi$. 

\begin{center}
\includegraphics*[scale=.5]{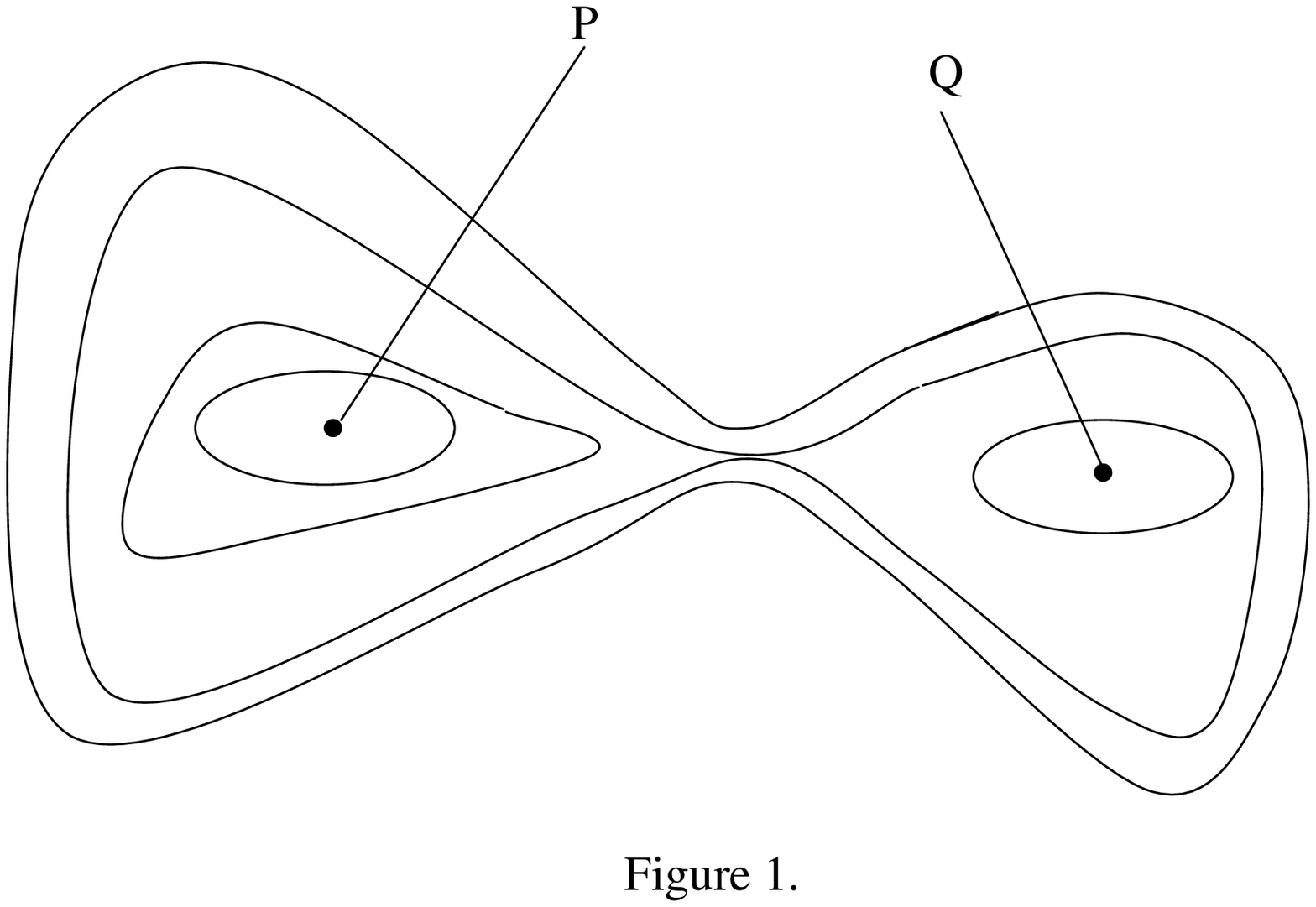}
\includegraphics*[scale=.4]{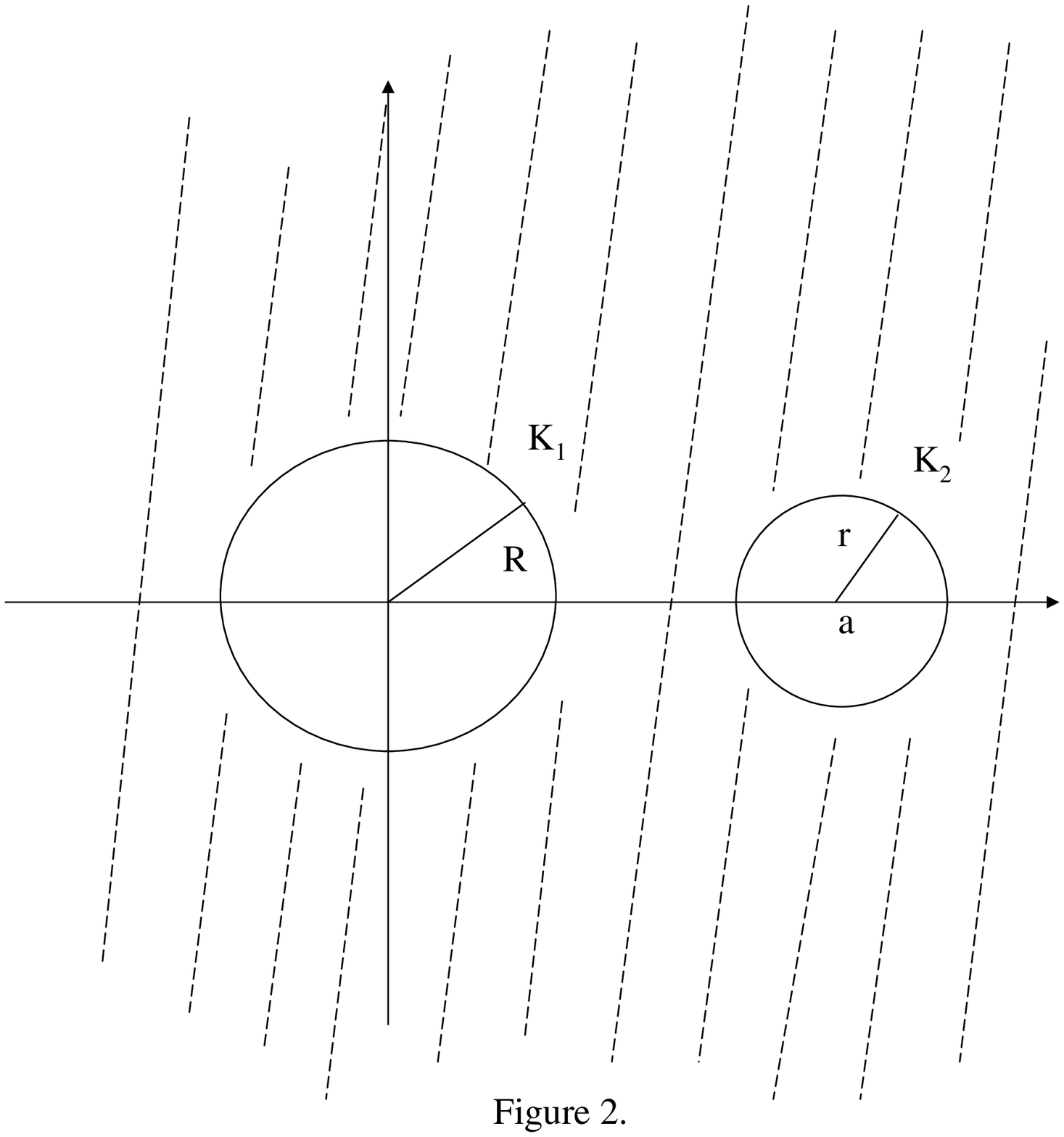}
\newpage
\includegraphics*[scale=.7]{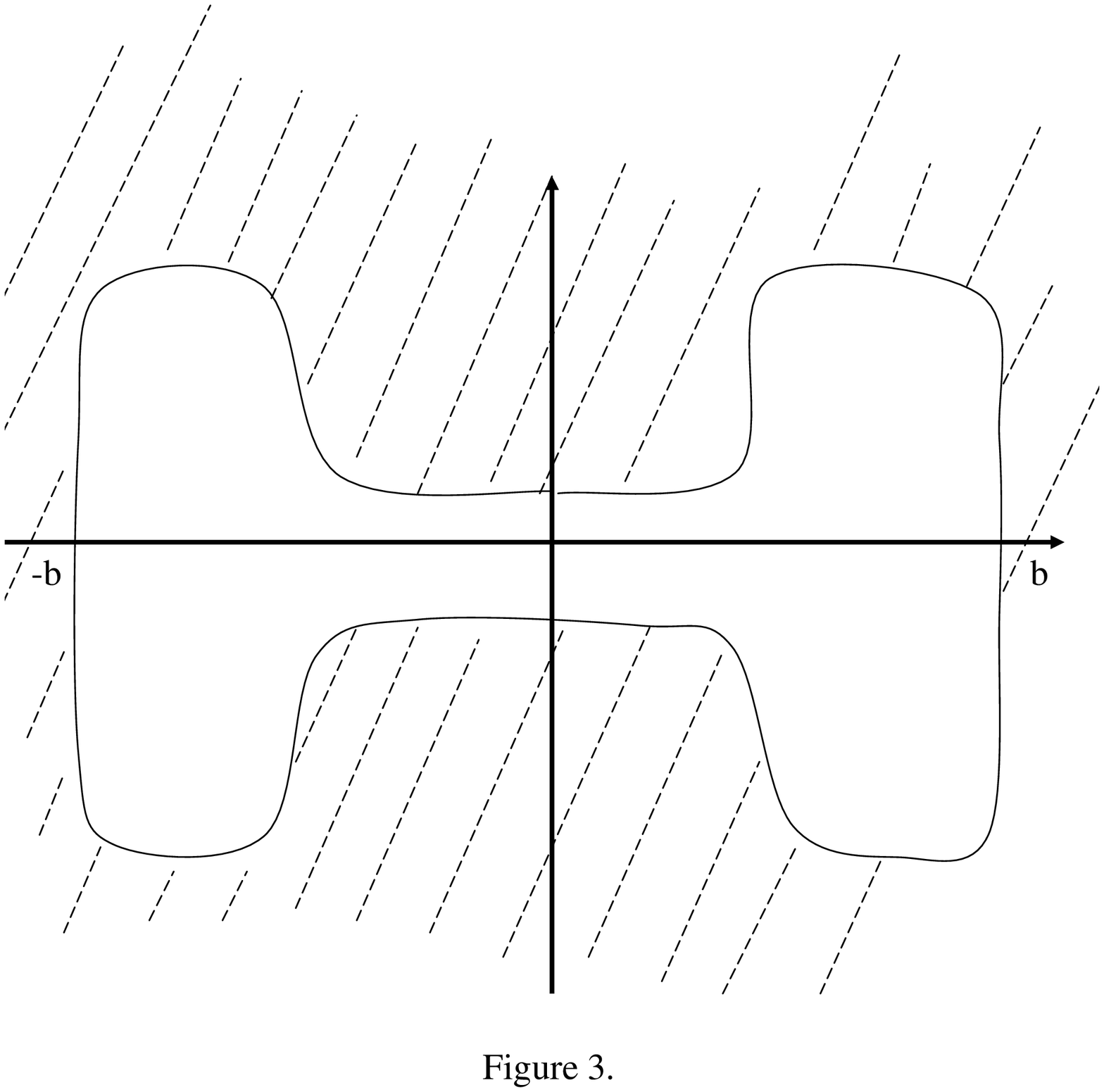}
\newpage
\includegraphics*[scale=.8]{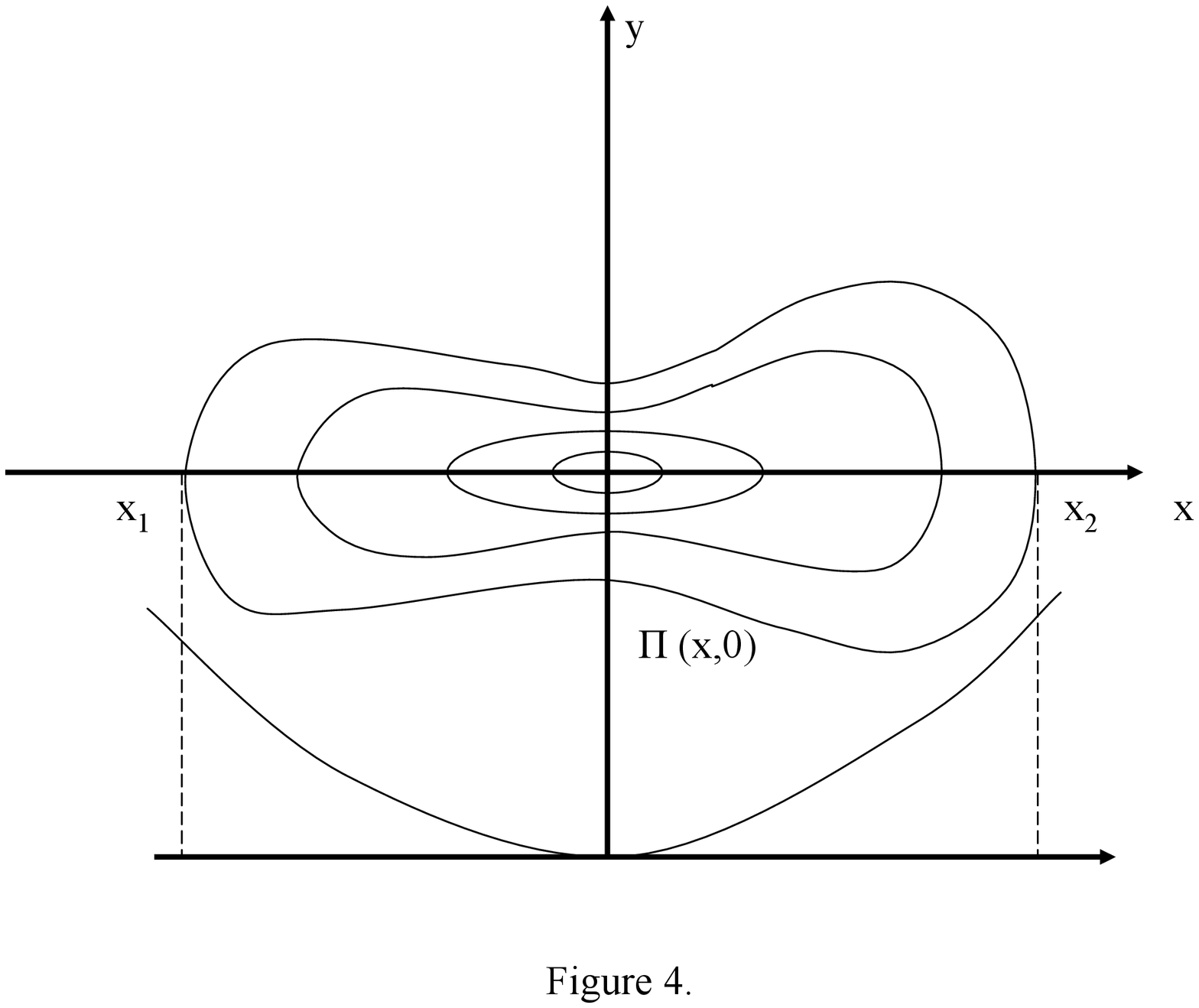}
\newpage
\includegraphics*[scale=.9]{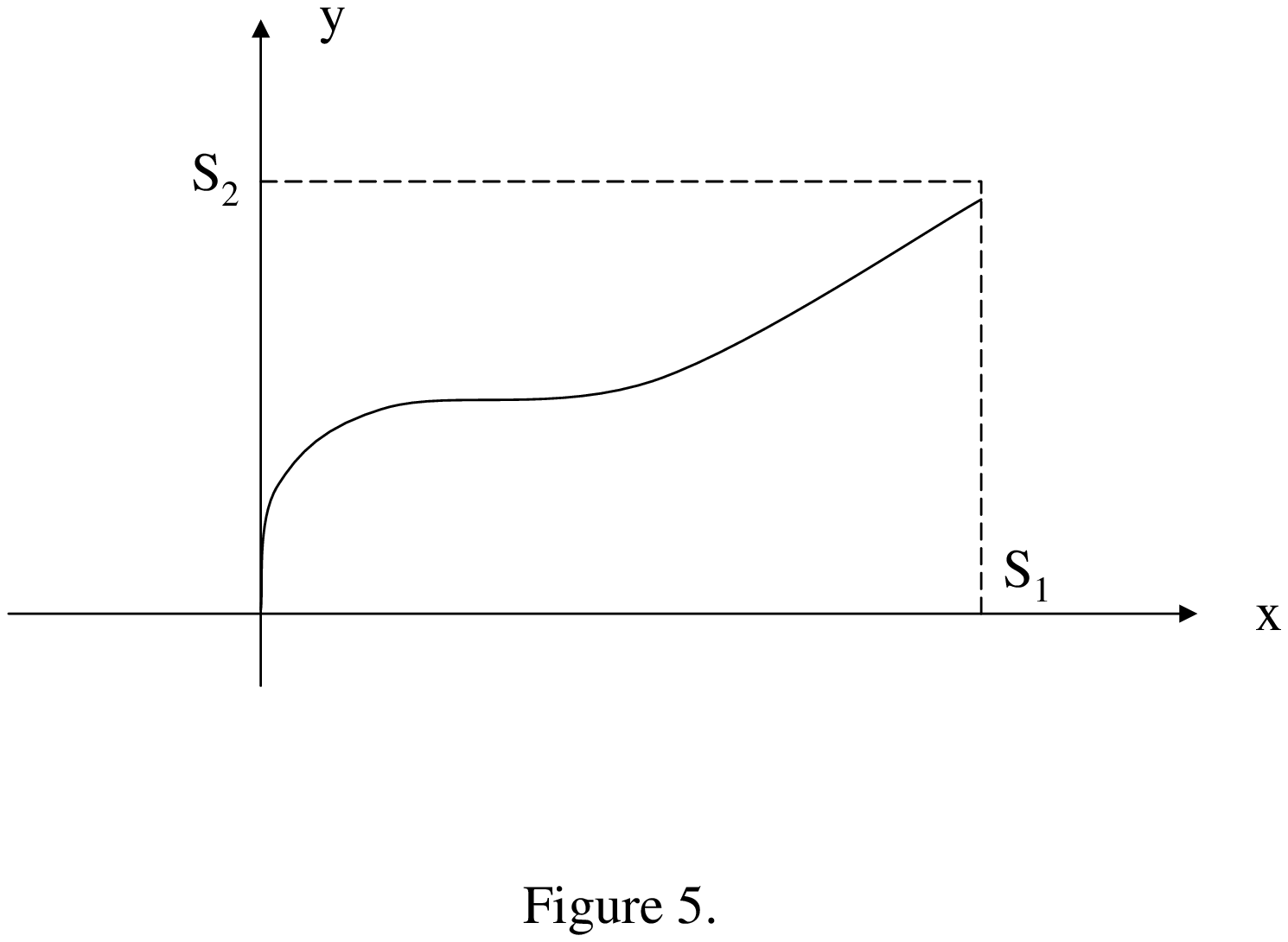}
\newpage
\includegraphics*[scale=.75]{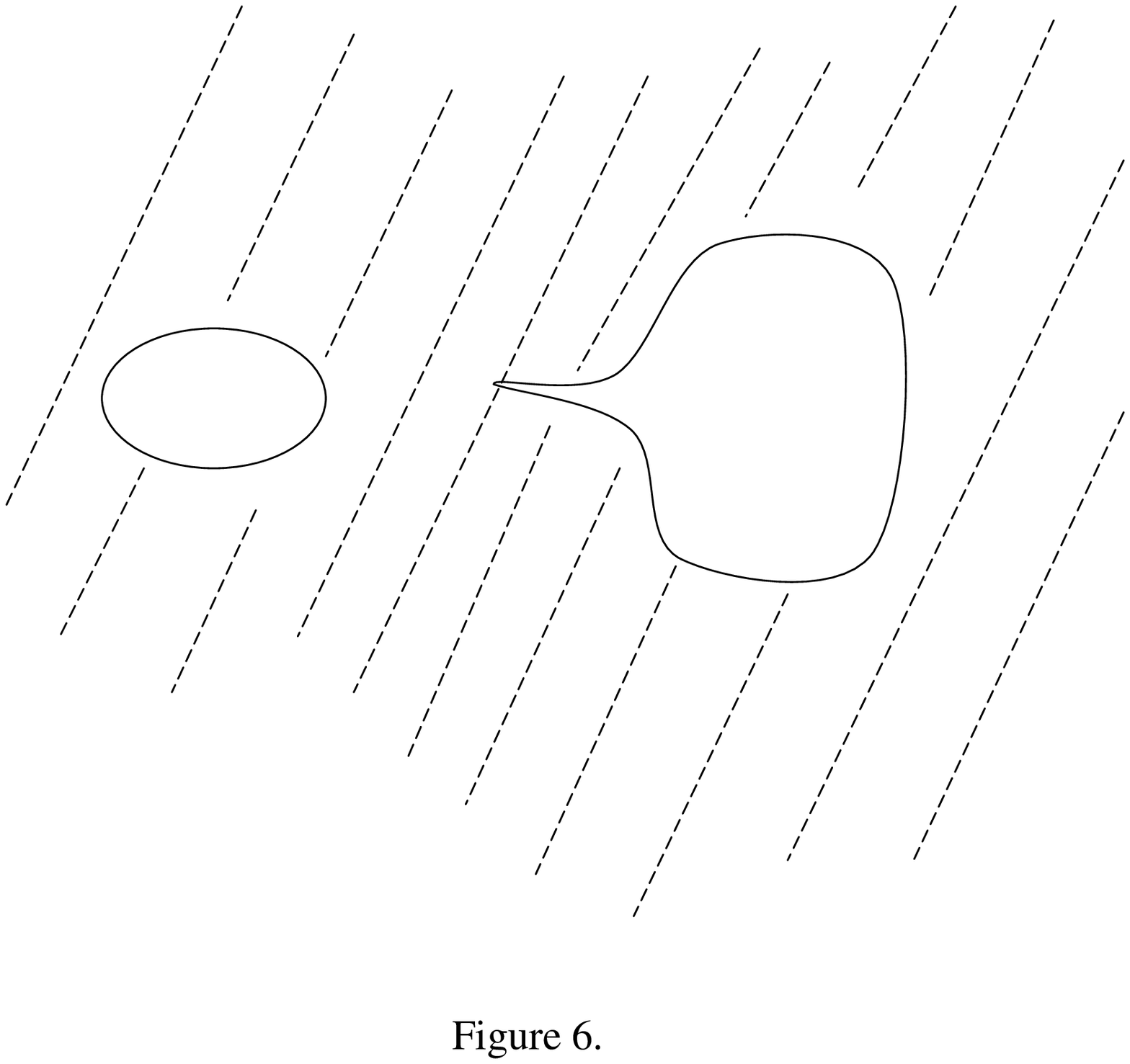}
\newpage
\includegraphics*[scale=.9]{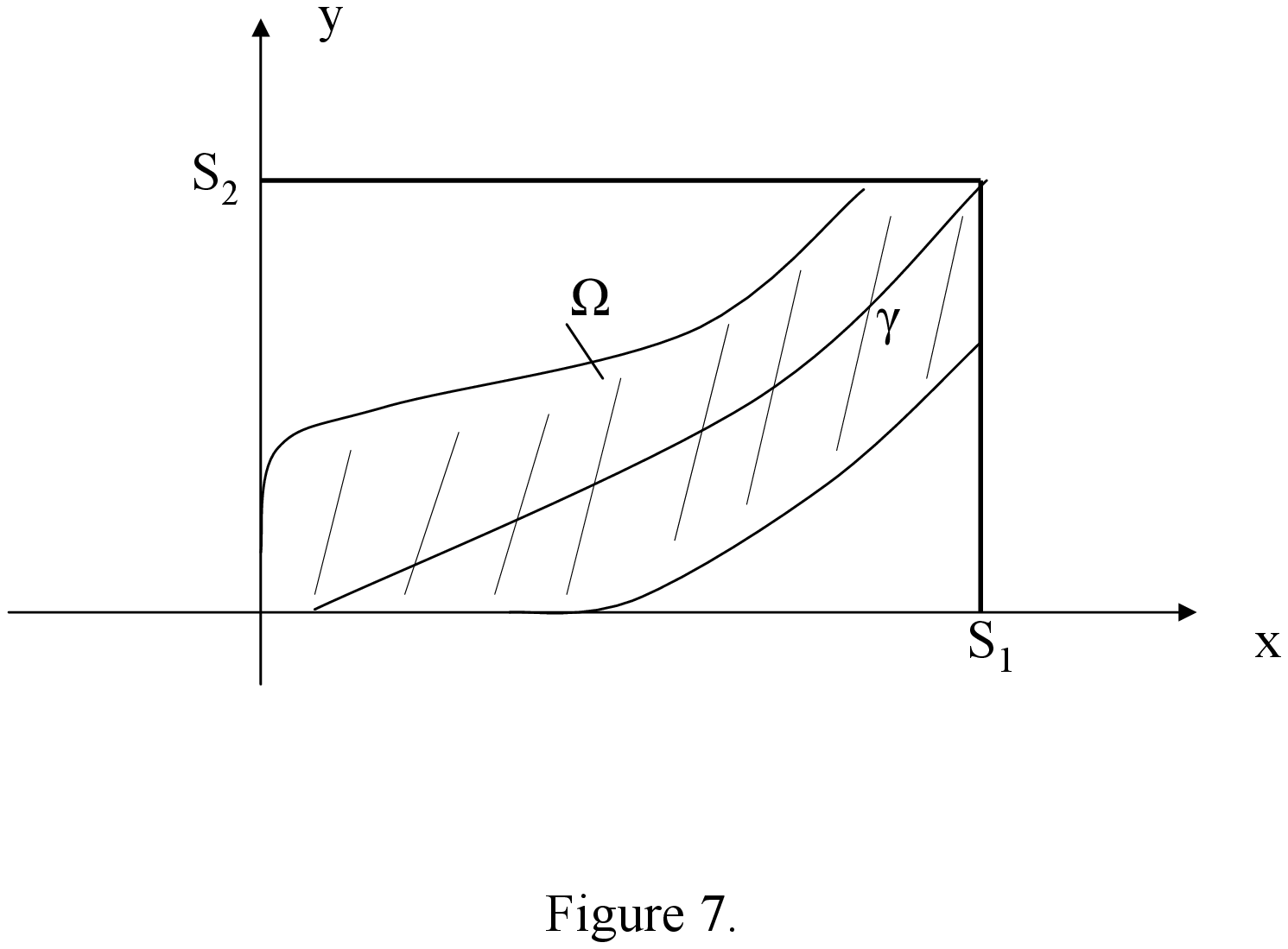}
\newpage
\includegraphics*[scale=.9]{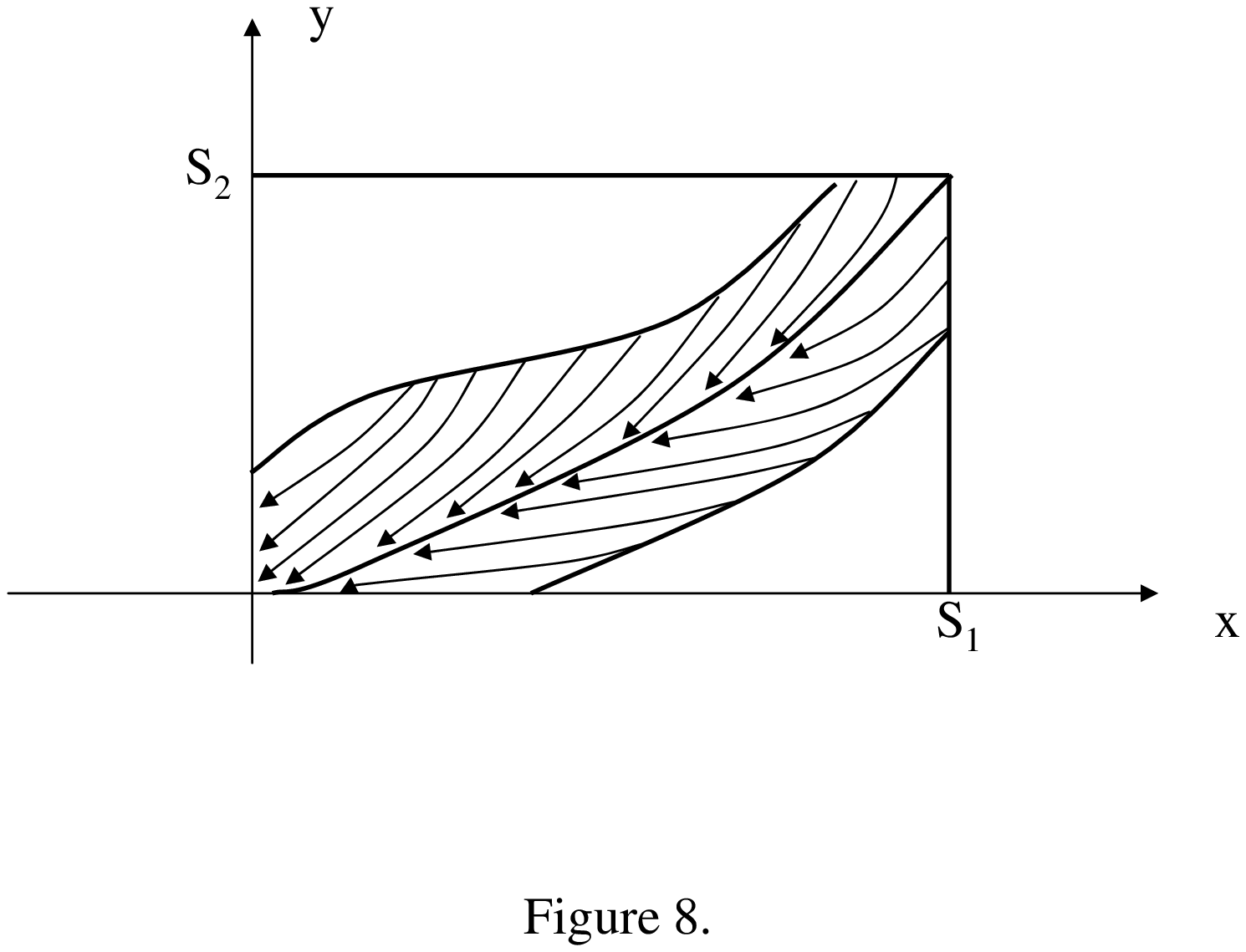}
\newpage
\includegraphics*[scale=0.9]{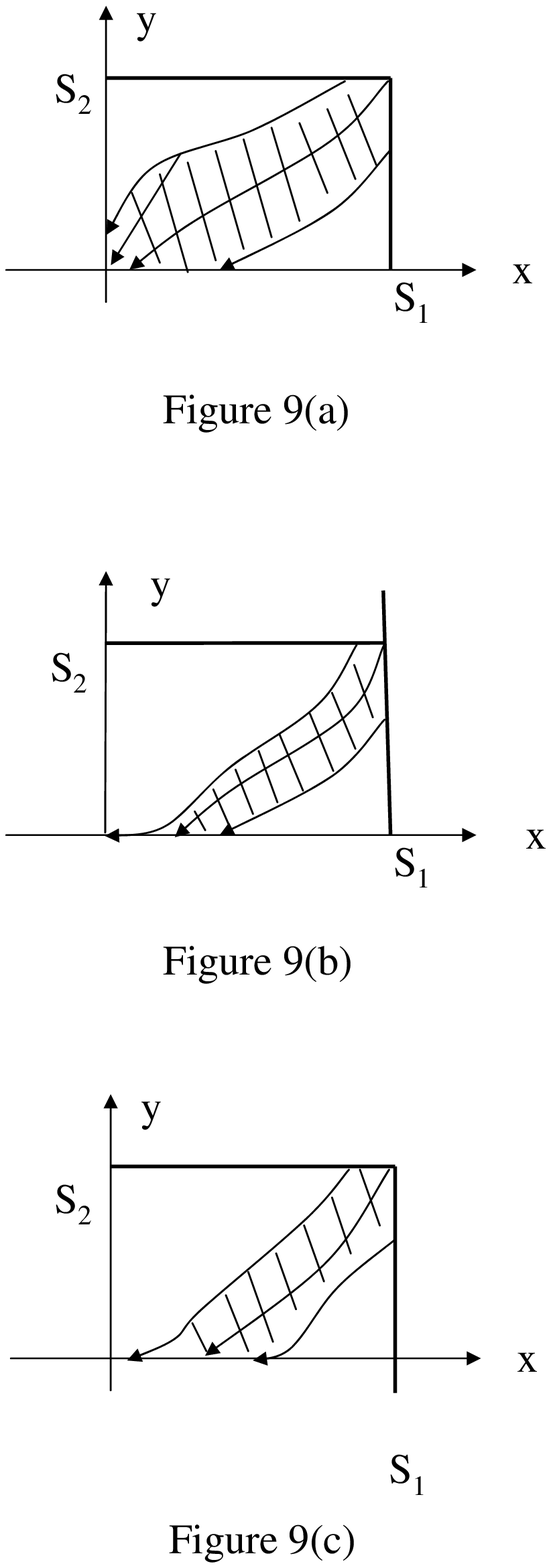}
\newpage
\includegraphics*[scale=.7]{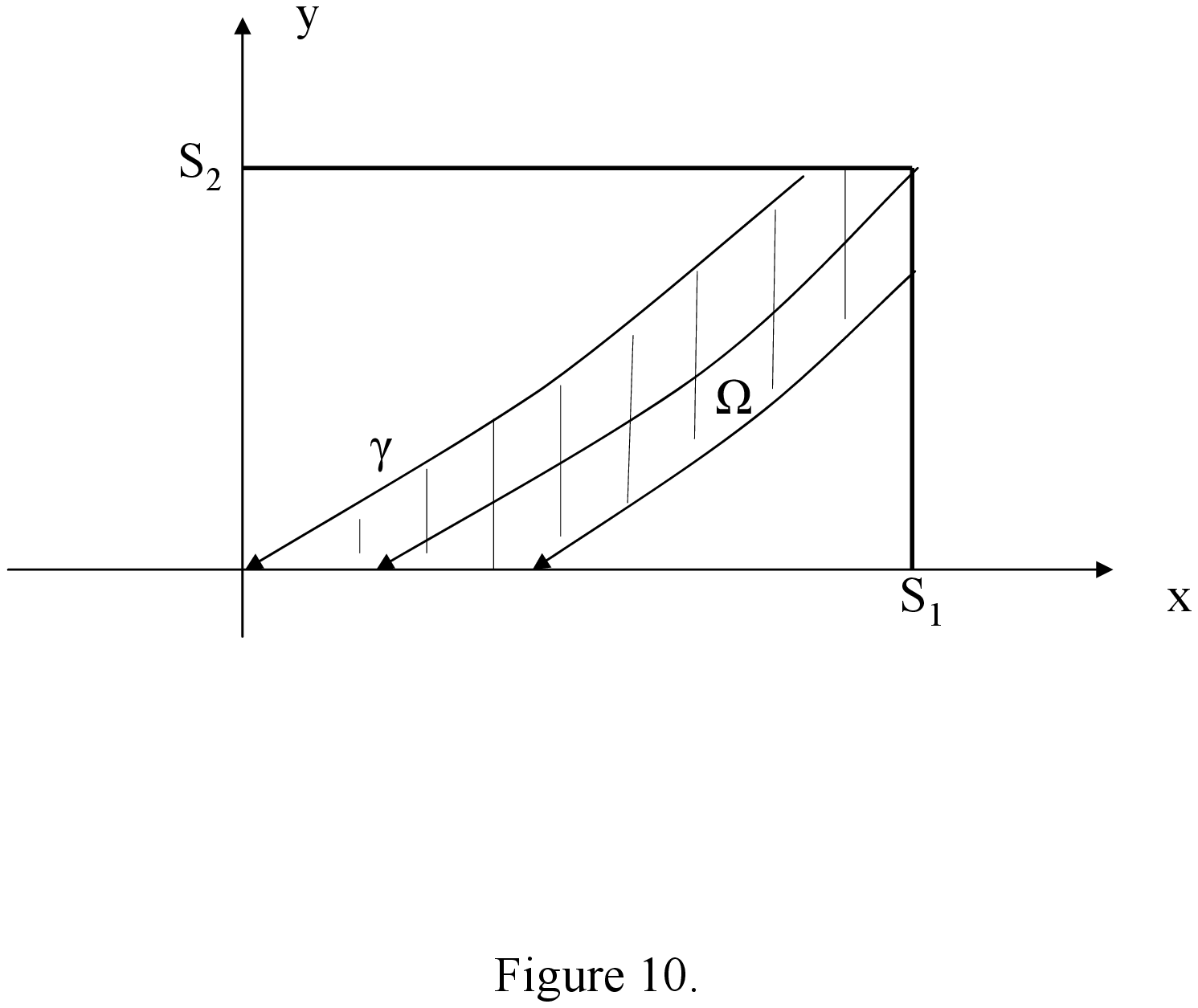}
\newpage
\includegraphics*[scale=.6]{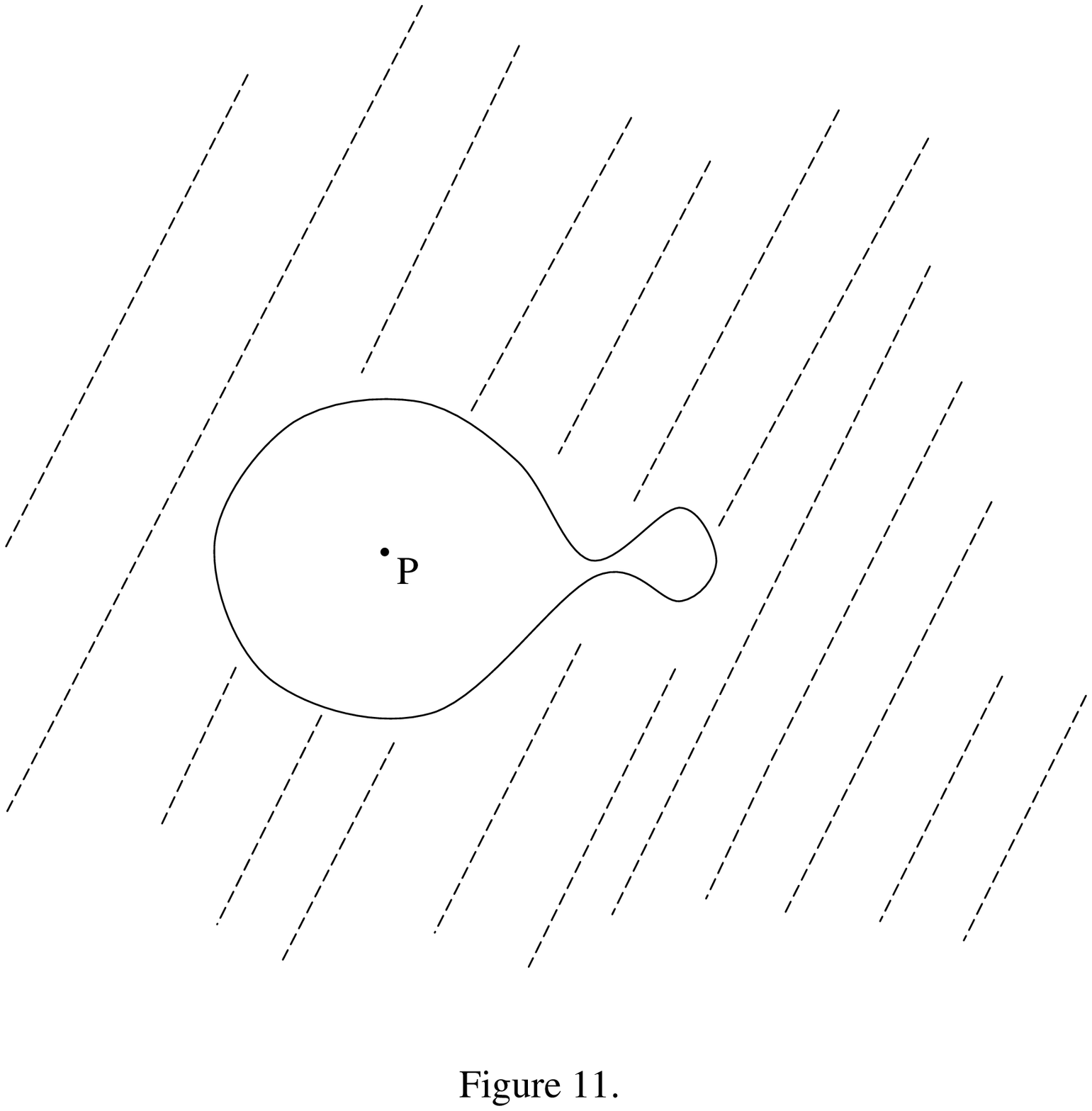}
\end{center}

\end{document}